\documentclass[letterpaper, conference]{ieeeconf}
\IEEEoverridecommandlockouts
\usepackage{graphics} 
\usepackage{epsfig} 
\usepackage{amsmath} 
\usepackage{amssymb}  
\usepackage{cite}
\usepackage{color}
\usepackage{graphicx}
\usepackage{algpseudocode}
\usepackage{algorithm}

\let\labelindent\relax
\usepackage{enumitem}
\renewcommand{\theenumi}{(\alph{enumi})}

\usepackage[font=footnotesize]{caption}

\newtheorem{theorem}{Theorem}[section]
\newtheorem{definition}[theorem]{Definition}
\newtheorem{assumption}[theorem]{Assumption}
\newtheorem{lemma}[theorem]{Lemma}

\newtheorem{remark}[theorem]{Remark}
\newtheorem{example}[theorem]{Example}

\newtheorem{proposition}[theorem]{Proposition}  
\newtheorem{problem}[theorem]{Problem}

\newcommand{\naturals}{\mathbb{N}}
\newcommand{\real}{\mathbb{R}}
\newcommand{\cplx}{\mathbb{C}}

\newcommand{\range}{\mathcal{R}}

\newcommand{\Fc}{\mathcal{F}}

\newcommand{\Kc}{\mathcal{K}}

\newcommand{\Mc}{\mathcal{M}}

\newcommand{\Pc}{\mathcal{P}}

\newcommand{\Sc}{\mathcal{S}}

\newcommand{\Vc}{\mathcal{V}}
\newcommand{\Wc}{\mathcal{W}}

\newcommand{\Pf}{\mathfrak{P}}

\newcommand{\EDMD}[3]{\operatorname{EDMD}(#1,#2,#3)}

\newcommand{\Kedmd}{K_{\operatorname{EDMD}}}

\newcommand{\until}[1]{\{1,\dots,#1\}}

\newcommand{\dx}{D(X)}
\newcommand{\dy}{D(Y)}

\newcommand{\tdx}{\tilde{D}(X)}
\newcommand{\tdy}{\tilde{D}(Y)}
\newcommand{\Span}{\operatorname{span}}

\newcommand{\ic}{\mathcal{I}_C}

\newcommand{\spec}{\operatorname{spec}}
\newcommand{\specnz}{\operatorname{spec}_{\neq 0}}
\newcommand{\sprad}{\operatorname{sprad}}
\newcommand{\rrmse}{\operatorname{RRMSE}}
\newcommand{\rrmsemax}{\operatorname{RRMSE_{\max}}}

\newcommand{\longthmtitle}[1]{\mbox{}{\textit{(#1):}}}
\newcommand{\setdef}[2]{\{#1 \; | \; #2\}}

\newcommand{\oprocendsymbol}{\hbox{$\square$}}
\newcommand{\oprocend}{\relax\ifmmode\else\unskip\hfill\fi\oprocendsymbol}
\def\eqoprocend{\tag*{$\square$}}

\renewcommand{\baselinestretch}{.975}

\allowdisplaybreaks

\title{\Large \bf Temporal Forward-Backward Consistency, Not Residual
Error, Measures the Prediction Accuracy of Extended Dynamic
Mode Decomposition
\thanks{This work was supported by
	ONR Award N00014-18-1-2828 and NSF Award IIS-2007141.}}

\author{Masih Haseli \quad Jorge Cort\'es
\thanks{The authors are with Department of Mechanical and Aerospace
Engineering, UC San Diego, {\tt\small
\{mhaseli,cortes\}@ucsd.edu}}
}

\graphicspath{{epsfiles/}}

\begin{document}

\maketitle

\begin{abstract}
Extended Dynamic Mode Decomposition (EDMD) is a popular data-driven
method to approximate the action of the Koopman operator on a linear
function space spanned by a dictionary of functions. The accuracy of
EDMD model critically depends on the quality of the particular
dictionary span\footnote{We consider the prediction accuracy
of EDMD for all (uncountably) functions in the dictionary span,
as opposed to only finitely many functions.}, specifically on
how close it is to being invariant under the Koopman operator.
Motivated by the observation that the residual error of EDMD,
typically used for dictionary learning, does not encode the quality
of the function space and is sensitive to the choice of basis, we
introduce the novel concept of consistency index.  We show that this
measure, based on using EDMD forward and backward in time, enjoys a
number of desirable qualities that make it suitable for data-driven
modeling of dynamical systems: it measures the quality of the
function space, it is invariant under the choice of basis, can be
computed in closed form from the data, and provides a tight
upper-bound for the relative root mean square error of all function
predictions on the entire span of the dictionary.
\end{abstract}

\section{Introduction}
Koopman operator theory has gained widespread attention in recent
years for the study of dynamical systems, chiefly thanks to the linear
structure of the operator despite nonlinearities in the system. In
fact, the linearity of the Koopman operator leads to computationally
efficient formulations to work with data.  This provides a
theoretically principled and explainable approach to the challenges
posed by incorporating data-driven methods in the modeling, analysis,
and control of dynamical systems.  Given that the Koopman operator is
generally infinite-dimensional, finding accurate finite-dimensional
approximations for its action is of utmost importance. The first step
to do so 
is to have a proper way to efficiently measure the quality of the
subspace. This measure in turn can be used for subspace identification
in optimization or neural network-based learning
methods. Defining such a measure is, in turn, the goal of this
paper.

\textit{Literature Review:} The Koopman operator~\cite{BOK:31}
provides an alternative representation of the evolution of dynamical
systems in terms of observables rather than system
trajectories. Despite possible nonlinearities in the system, the
eigenfunctions of the Koopman operator evaluated on the system's
trajectories have linear temporal evolution, and this leads to
efficient numerical methods used in complex system
analysis~\cite{IM:05,SPN-SS-EY:20}, estimation~\cite{MN-LM:18},
control~\cite{MK-IM-automatica:18,SP-SK:17,DG-DAP:21,VZ-EB:22-arxiv},
and robotics~\cite{GM-MLC-XT-TDM:21,LS-KK:21}, to name a few. Despite
these appealing applications, the infinite-dimensional nature of the
operator makes its direct use on digital computers challenging.
Dynamic Mode Decomposition (DMD)~\cite{PJS:10} and its generalization
Extended Dynamic Mode Decomposition (EDMD)~\cite{MOW-IGK-CWR:15} are
popular data-driven methods to approximate the action of the Koopman
operator on finite-dimensional spaces. EDMD in particular uses a
dictionary of functions whose span specifies the finite-dimensional
space of choice.  The work in~\cite{HL-DMT:20} studies the prediction
accuracy of DMD while~\cite{MK-IM:18} provides several convergence
results for EDMD as the number of data points and dimension of the
dictionary go to infinity.  The dependence of the EDMD's prediction
accuracy on the choice of the dictionary has led to a search for
dictionaries whose span is close to being invariant under the Koopman
operator~\cite{SLB-BWB-JLP-JNK:16}. The works
in~\cite{QL-FD-EMB-IGK:17,NT-YK-TY:17} use methods based on neural
networks for this task, while~\cite{MK-IM:20} directly learns the
Koopman eigenfunctions spanning invariant subspaces. Moreover, the
works in~\cite{PB-MB-SK-AL-SS-SH:21,FF-BY-DR-GS-IRM:21} approximate
finite-dimensional Koopman models relying on knowledge about the
system's attractors and their stability.  In our previous work, we
have provided efficient algebraic algorithms to identify exact
Koopman-invariant subspaces~\cite{MH-JC:22-tac,MH-JC:21-tcns} or
approximate them with tunable predefined
accuracy~\cite{MH-JC:21-auto}.

\textit{Statement of Contributions:} Our starting point\footnote{
We use the following notation. The symbols $\naturals$, $\real$, and
$\cplx$, represent the sets of natural, real, and complex numbers
resp.  Given $A \in \cplx^{m \times n}$, we denote its
transpose, pseudo-inverse, conjugate transpose, Frobenius norm and
range space by $A^T$, $A^\dagger$, $A^H$, $\|A\|_F$, and $\range(A)$
resp. If $A$ is square, we use $A^{-1}$ to denote its
inverse. We denote by $\spec(A)$ the spectrum of $A$. Similarly,
$\specnz(A)$ denotes the set of nonzero eigenvalues of
$A$. Moreover,
$\sprad(A) := \max \setdef{|\lambda|}{\lambda \in \spec(A)}$ is the
spectral radius of $A$.  If $\spec(A) \subset \real$, then
$\lambda_{\min}(A)$ and $\lambda_{\max}(A)$ denote the smallest and
largest eigenvalues of~$A$.  We use $I_m$ and
$\mathbf{0}_{m\times n}$ to denote the $m \times m$ identity matrix
and $m\times n$ zero matrix (we drop the indices when
appropriate). We denote the $2$-norm of the vector $v \in \cplx^n$
by $\|v\|_2$. Given sets $S_1$ and $S_2$, their union and
intersection are represented by $S_1 \cup S_2$ and $S_1 \cap
S_2$. Also, $S_1 \subseteq S_2$ and $S_1 \subsetneq S_2$
resp. mean that $S_1$ is a subset and proper subset of
$S_2$. Given the vector space $\Vc$ defined on the field $\cplx$,
$\dim \Vc$ denotes its dimension. Moreover, given a set
$\Sc \subseteq \Vc$, $\Span(\Sc)$ is a vector space comprised of all
linear combinations of elements in $\Sc$. If vectors
$v, w \in \real^n$ and vector spaces $\Vc, \Wc \subseteq \real^n$
are orthogonal, we write $v \perp w$ and $\Vc \perp \Wc$. Moreover,
$\Vc ^\perp$ denotes the orthogonal complement of $\Vc$. Given
functions $f$ and $g$ with appropriate domains and co-domains,
$f \circ g$ denotes their composition.}  is the
observation that the residual error of EDMD, typically used for
dictionary learning, does not necessarily measure the quality of the
subspace spanned by the dictionary
and, consequently, the prediction accuracy of EDMD on the subspace.
To illustrate this point, we provide an example showing that one can
choose a sequence of dictionaries spanning the same subspace that make
the residual error arbitrarily close to zero.  This motivates our goal
of identifying better measures to assess the EDMD's prediction
accuracy and its dictionary's quality.  We define the notion of the
consistency matrix and its spectral radius, which we term consistency
index, which measures the deviation of the EDMD solutions forward and
backward in time from being the inverse of each other. This is
justified by the fact that if a subspace is Koopman invariant, the
EDMD solutions applied forward and backward in time are the inverse of
each other. We characterize various algebraic properties of the
consistency index and show that it only depends on the data and the
space spanned by the dictionary, and is hence invariant under changes
of basis. We also establish that the square root of the consistency
index provides a tight upper bound on the relative root mean square
EDMD prediction error of all functions in the dictionary's~span.

\section{Preliminaries}\label{sec:preliminaries}
We briefly recall basic facts about the Koopman
operator~\cite{MB-RM-IM:12} and Extended Dynamic Mode
Decomposition~\cite{MOW-IGK-CWR:15}.

\subsubsection*{Koopman Operator}
Consider a dynamical system with state space $\Mc \subseteq \real^n$
\begin{align}\label{eq:dymamical-sys}
x^+ = T(x), \; x \in \Mc.
\end{align}
Let $\Fc$ be a vector space defined on $\cplx$ comprised of functions
from $\Mc$ to $\cplx$ whose composition with $T$ also belong to~$\Fc$.
The Koopman operator associated with the dynamics is
\begin{align}\label{eq:Koopman-def}
\Kc f = f \circ T.
\end{align}
The operator is linear. Its eigenfunctions have linear evolution on
the trajectories of the system, i.e., given eigenfunction $\phi$ with
eigenvalue $\lambda$, $\phi(x^+) = \phi \circ T(x) = \Kc \phi(x) = \lambda \phi(x)$. This results
in a significant property of the Koopman eigendecomposition: given
eigenpairs $\{(\lambda_i, \phi_i)\}_{i = 1}^{N_k}$, the evolution of
$f = \sum_{i=1}^{N_k} c_i \phi_i$ on a trajectory of the system
starting from $x_0 \in \Mc$ is
$f(x(k)) = \sum_{i=1}^{N_k} c_i \lambda_i^k \, \phi_i(x_0)$, for
$ k \in \naturals$.
This linear property is useful for both prediction and identification,
as the linearity always holds even if the system is
nonlinear. However, to completely capture the dynamics, one might need
the space $\Fc$ to be infinite dimensional.

\subsubsection*{Extended Dynamic Mode Decomposition}
The infinite-dimensional property of the Koopman operator prevents its
direct use in practical data-driven settings. This leads naturally to
constructing finite-dimensional approximations, e.g., using Extended
Dynamic Mode Decomposition (EDMD).  EDMD uses a dictionary
$D: \Mc \to \real^{1 \times N_d}$ containing $N_d$ functions,
$D(x) = [d_1(x), \ldots, d_{N_d}(x)]$.
To capture the dynamic behavior,
EDMD uses data $X,Y \in \real^{N \times n}$ containing $N$ data
snapshots gathered from system trajectories,
\begin{align}\label{eq:data-snapshots}
y_i = T(x_i), \, \forall i \in \until{N},
\end{align}
where $x_i^T$ and $y_i^T$ correspond to the $i$th rows of $X$
and~$Y$. For convenience, we define the action of $D$ on a data
matrix as $\dx = [D(x_1)^T, D(x_2)^T, \ldots, D(x_N)^T]^T$, where
$x_i^T$ is the $i$th row of $X$. EDMD approximates the action of
the Koopman operator on $\Span(D)$ by solving
\begin{align}\label{eq:EDMD-optimization}
\underset{K}{\text{minimize}} \| D(Y) - D(X) K \|_F
\end{align} 
which has the closed-form solution
\begin{align}\label{eq:EDMD-closed-form}
\Kedmd =  \EDMD{D}{X}{Y} := \dx^\dagger \dy.
\end{align}
We rely on the following basic assumption.

\begin{assumption}\longthmtitle{Full Rank Dictionary
Matrices}\label{a:D(X)-full-rank}
$\dx$ and $\dy$ have full column rank.  \oprocend
\end{assumption}

Assumption~\ref{a:D(X)-full-rank} implies that the functions in $D$
are linearly independent, i.e., they form a basis for $\Span(D)$ and
the data are diverse enough to distinguish between the elements of
$D$. Assumption~\ref{a:D(X)-full-rank} ensures that $\Kedmd$ is the
unique solution for~\eqref{eq:EDMD-optimization}.
One can use $\Kedmd$ to approximate the Koopman eigenfunctions and,
more importantly, the action of the operator on $\Span(D)$. Given
$f \in \Span(D)$ in the form of $f(\cdot) = D(\cdot)v_f$ for
$v_f \in \cplx^{N_d}$, one defines the EDMD predictor function for
$\Kc f$ as
\begin{align}\label{eq:EDMD-Koopman-pred}
\Pf_{\Kc f} (\cdot) = D(\cdot) \Kedmd v_f.
\end{align}
The predictor's quality depends on the quality of the dictionary. If
$\Span(D)$ is Koopman-invariant (i.e., $\Kc g \in \Span(D)$ for all
$g \in \Span(D)$), the predictor~\eqref{eq:EDMD-Koopman-pred} is exact
(otherwise, the prediction is inexact for some functions in the
space).

\section{Motivation and Problem
Statement}\label{sec:problem-statement}
The quality of the dictionary used for EDMD directly impacts its
accuracy. Since in general the dynamics is unknown, it is important to
use data to learn a proper dictionary tailored to the dynamics.  The
\emph{residual error} of EDMD, $\|\dy - \dx \Kedmd\|_F$, is commonly
used for this purpose as an objective function in optimization and
neural network-based learning schemes. However, it is important to
note that even though a high quality subspace (close to
Koopman-invariant) leads to small residual error, the converse is not
true: a dictionary $D$ with small residual error does not necessarily
mean that EDMD's prediction is accurate on $\Span(D)$.

\begin{example}\longthmtitle{Residual Error is not Invariant under
Linear Transformation of  Dictionary}\label{ex:motivation}
Consider the linear system $x^+ = 0.5 x$ and the vector space of
functions $\Sc = \Span\{x, x^3 - x^2\}$. To apply EDMD, we gather
$N = 1000$ data snapshots from trajectories of the system with
length of two time steps and initial conditions uniformly selected
from $[-2,2]$, and form data matrices $X,Y \in \real^{N \times
1}$. We consider a family of dictionaries parameterized by
$\alpha \in \real \setminus \{0\}$,
\begin{align}\label{eq:dictionary-alpha}
D_\alpha(x) = [x, x + \alpha (x^3 - x^2)].
\end{align}
Note that each $ D_\alpha$ is a basis for $\Sc$, and all the
dictionaries are related by nonsingular linear transformations.  We
also define two notions of prediction accuracy: the residual error
$E$ of EDMD and its normalized version $E_{\operatorname{rel}}$,  
\begin{align*}
&E (\alpha) = \| D_\alpha(Y) - D_\alpha(X) K_{\alpha}  \|_F,
&&E_{\operatorname{rel}} (\alpha) = \frac{E (\alpha)}{\| D_\alpha(Y) \|_F}, 
\end{align*}
where $K_{\alpha} = \EDMD{D_\alpha}{X}{Y}$.

\begin{figure}[htb]
\centering 
{\includegraphics[width=.48\linewidth]{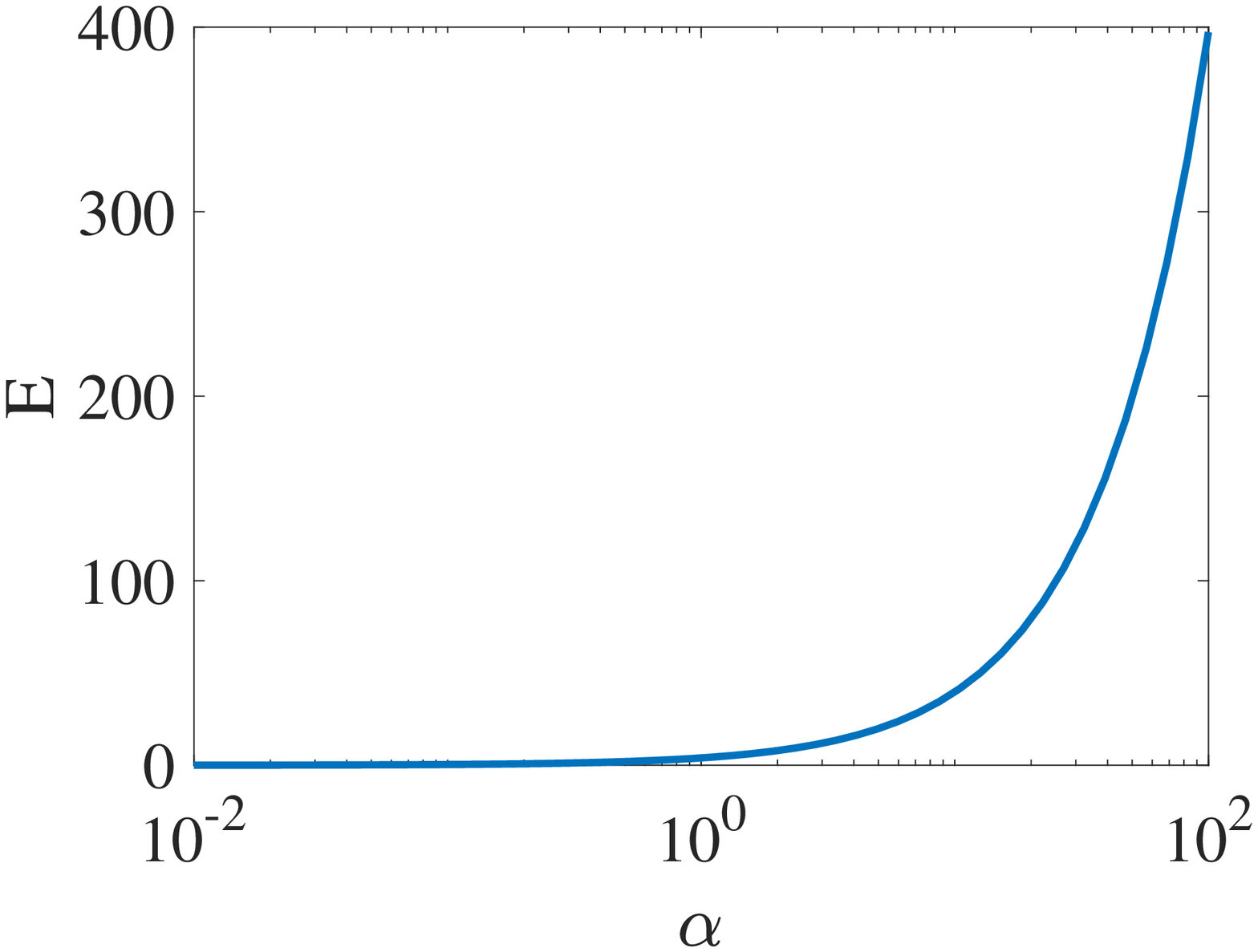}}
{\includegraphics[width=.48\linewidth]{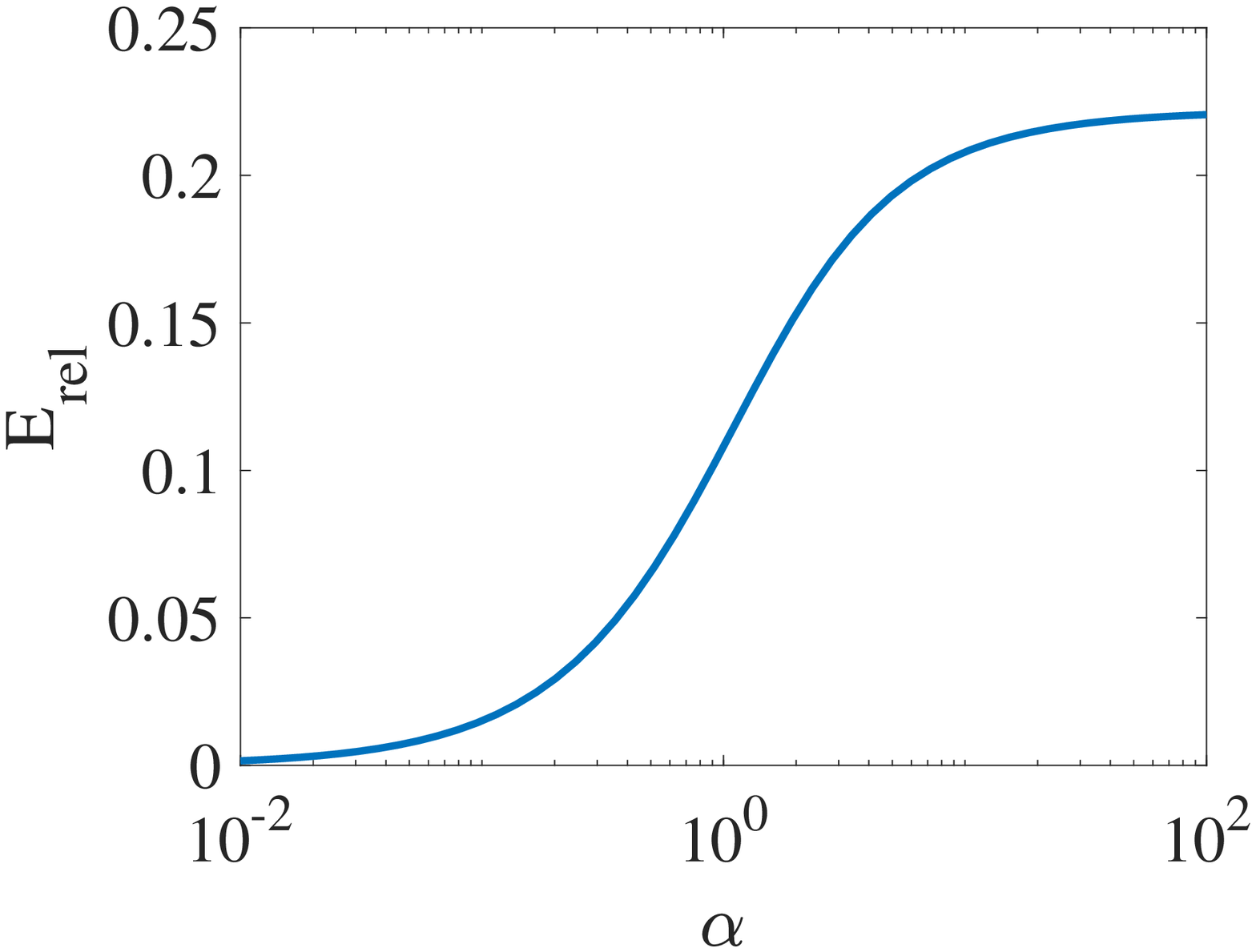}}
\caption{EDMD's residual error (left) and relative residual error
(right) as a function of $\alpha$
in~\eqref{eq:dictionary-alpha}.}\label{fig:error-vs-alpha}
\vspace*{-2ex}
\end{figure}  
Figure~\ref{fig:error-vs-alpha} shows the aforementioned notions of
error versus the value of $\alpha \in [0.01,
100]$. Figure~\ref{fig:error-vs-alpha} clearly demonstrates the
sensitivity of errors to the choice of basis for $\Sc$ despite the
invariance of $\spec(K_{\operatorname{EDMD},\alpha})$ under the choice
of basis (see e.g.,~\cite[Lemma~7.1]{MH-JC:21-auto}).  By tuning
$\alpha$, one can make both errors arbitrarily close to zero. As a
result, using the residual error as a measure to assess the quality of
the space or as an objective function in optimization or neural
network-based dictionary learning schemes can lead to erroneous
results.  \oprocend
\end{example}

\begin{remark}\longthmtitle{Prediction Accuracy of Dictionary
Elements versus Dictionary's Span}
Some applications only require short-term prediction of
\textit{finitely} many observables. In such cases, the residual
error of EDMD may be useful: the observables of choice are fixed
as elements of the dictionary and the rest of the functions
composing the dictionary are learned by minimizing the residual
error (which measures the average one time-step prediction
error). However, such methods do not necessarily lead to a
dictionary which spans an approximate Koopman-invariant subspace.
\oprocend
\end{remark}

The observations in Example~\ref{ex:motivation} prompts us to search
for a better measure of the dictionary's quality and therefore the
EDMD's prediction accuracy. We formalize this next.

\begin{problem}\longthmtitle{Characterization of EDMD's Prediction
Accuracy and the Dictionary's
Quality}\label{prob:dictionary-invariance-measure}
Given a dictionary $D$ and data matrices $X$ and $Y$, under
Assumption~\ref{a:D(X)-full-rank}, we aim to provide a data-driven measure of EDMD's accuracy and the dictionary's quality that
\begin{enumerate}
\item only depends on $\Span(D)$, $X$, and $Y$, and hence is 
invariant under the choice of basis for $\Span(D)$, i.e., given
$D'$ as an alternative basis for $\Span(D)$, the accuracy
measures calculated based on $D$ and $D'$ are~equal;
\item provides a data-driven bound on the distance between $\Kc f$
and its EDMD prediction $\Pf_{\Kc f}$ for all functions
$ f \in \Span(D)$;
\item can be computed using a closed-form formula (for
implementation in optimization solvers).  \oprocend
\end{enumerate}
\end{problem}

\section{Temporal Forward-Backward Consistency}\label{sec:consistency}
Here, we take the first step towards finding an appropriate measure
for EDMD's prediction accuracy by comparing the solutions of EDMD
forward and backward in time\footnote{The idea of looking forward and
backward in time has been considered in the literature for different
purposes, such as improving DMD to deal with noisy
data~\cite{STMD-MSH-MOW-CWR:16,OA-WY-AB:19} and identifying exact
Koopman eigenfunctions in our previous work~\cite{MH-JC:22-tac} but,
to the best of our knowledge, not for formally characterizing EDMD's
prediction accuracy.}.  Throughout the paper, we use the following
notation for forward and backward EDMD matrices
\begin{subequations}\label{eq:fwd-bck-EDMD-matrices}
\begin{align}
K_F = \EDMD{D}{X}{Y} = \dx^\dagger \dy,
\\
K_B = \EDMD{D}{Y}{X} = \dy^\dagger \dx.
\end{align}  
\end{subequations}
We rely on the observation that if the dictionary spans a
Koopman-invariant subspace, then $K_F K_B = I$. Otherwise, the forward
and backward EDMD matrices will not be the inverse of each other,
which motivates the next definition.

\begin{definition}\longthmtitle{Consistency
Matrix and Index}\label{def:consistency-matrix}
Given dictionary $D$ and data matrices $X$ and $Y$, the
\emph{consistency matrix} is $M_C(D,X,Y)= I - K_F K_B$ and the
\emph{consistency index} is
$ \mathcal{I}_C (D,X,Y) = \sprad \big(M_C(D,X,Y)\big)$.  \oprocend
\end{definition}

For convenience, we refer to $M_C(D,X,Y)$ and $\ic(D,X,Y)$ as $M_C$
and $\ic$ when the context is clear.  Next, we show that the
eigenvalues of the consistency matrix are invariant under linear
transformations of the dictionary.

\begin{proposition}\longthmtitle{Consistency Matrix's Spectrum is
Invariant under Linear Transformation of
Dictionary}\label{p:spectrum-invariance-consistency}
Let $D$ and $\tilde{D}$ be two dictionaries such that
$\tilde{D}(\cdot) = D(\cdot) R$, where $R$ is an invertible
matrix. Moreover, given data matrices $X$ and $Y$ let
Assumption~\ref{a:D(X)-full-rank} hold. Then,
\begin{enumerate}
\item $M_C(\tilde{D},X,Y) =R^{-1} M_C(D,X,Y) R$;
\item $\spec\big(M_C(D,X,Y)\big) = \spec\big(M_C(\tilde{D},X,Y)\big)$.
\end{enumerate}
\end{proposition}
\begin{proof}
Note that part (b) directly follows from part~(a) and the fact that
similarity transformations preserve the eigenvalues. To show part
(a), define for convenience,
\begin{align*}
&K_F = \dx^\dagger \dy,  &&K_B = \dy^\dagger \dx,
\\
&\tilde{K}_F = \tdx^\dagger \tdy,  &&\tilde{K}_B = \tdy^\dagger \tdx.
\end{align*}
We start by showing that $K_F K_B$ and $\tilde{K}_F \tilde{K}_B$ are
similar. By definition, one can write
\begin{align}\label{eq:forward-backward-second-coordinate}
\tilde{K}_F \tilde{K}_B = \tdx^\dagger \tdy \tdy^\dagger \tdx.
\end{align}
Moreover, given Assumption~\ref{a:D(X)-full-rank} and the definition
of $\tilde{D}$,
\begin{align}\label{eq:psuedoinverse-coordinate-change}
&\tdx^\dagger = \big(\tdx^T \tdx \big)^{-1} \tdx^T \nonumber \\
&=\big(R^T \dx^T \dx R \big)^{-1} R^T \dx^T = R^{-1} \dx^\dagger.
\end{align}
Equations~\eqref{eq:forward-backward-second-coordinate}-\eqref{eq:psuedoinverse-coordinate-change},
in conjunction with the fact that
$\tdy \tdy^\dagger  =  \dy \dy^\dagger$
(cf. Lemma~\ref{l:invariance-projection-change-coordinates} in the
appendix), imply $\tilde{K}_F \tilde{K}_B = R^{-1} K_F K_B R$
directly leading to the required identity following
Definition~\ref{def:consistency-matrix}.
\end{proof}

According to Proposition~\ref{p:spectrum-invariance-consistency}, the
spectrum of the consistency matrix is a property of the data and the
\emph{vector space} spanned by the dictionary, as opposed to the
dictionary itself. This property is consistent with the requirement in
Problem~\ref{prob:dictionary-invariance-measure}(a). Next, we further
investigate the eigendecomposition of the consistency matrix.

\begin{lemma}\longthmtitle{Consistency Matrix's
Properties}\label{l:consistency-matrix-properties}
Given Assumption~\ref{a:D(X)-full-rank}, the consistency matrix
$M_C(D,X,Y)$ satisfies:
\begin{enumerate}
\item it is similar to a symmetric matrix;
\item it is diagonalizable with a complete set of eigenvectors;
\item $\spec\big( M_C(D,X,Y) \big) \subset [0,1]$.
\end{enumerate}
\end{lemma}
\begin{proof}
(a) Given Assumption~\ref{a:D(X)-full-rank}, there exists an
invertible matrix $R$ such that the columns of $\dx R$ are
orthonormal. Define the dictionary $\tilde{D}(\cdot) = D(\cdot)
R$. Note that $\tdx^T \tdx = I_{N_d}$ and hence
$\tdx^\dagger = \big(\tdx^T \tdx \big)^{-1} \tdx^T = \tdx^T$.
Hence,
\begin{align*}
M_C(\tilde{D}, X,Y) = I - \tdx^T \tdy \tdy^\dagger \tdx. 
\end{align*}
Noting that $\tdy \tdy^\dagger$ is symmetric, we deduce that
$M_C(\tilde{D}, X,Y)$ is symmetric. Then (a) directly follows by the
definition of $R$ and
Proposition~\ref{p:spectrum-invariance-consistency}(a).

(b) The proof directly follows from part~(a) and the fact that
symmetric matrices are diagonalizable and have a complete set of
eigenvectors.

(c) From part~(a), we deduce that $M_C(D,X,Y)$ has real
eigenvalues. Since
$ M_C(D, X,Y) = I - \dx^\dagger \dy \dy^\dagger \dx$, we only need
to show
\begin{align}\label{eq:mspec-to-prove}
\spec\big(\dx^\dagger \dy \dy^\dagger \dx \big) \subset [0,1].
\end{align}
Consider an eigenvector $v \in \real^{N_d}\setminus \{0\}$ with
eigenvalue $\mu$, i.e., $\dx^\dagger \dy \dy^\dagger \dx v = \mu v$.
Multiplying both sides from the left by $\dx$ and defining
$w = \dx v$ leads to $\dx \dx^\dagger \dy \dy^\dagger w = \mu w$.
Next, multiplying this equation from the left by $w^T$,
\begin{align}\label{eq:rayleigh-quotient-1}
w^T \dx \dx^\dagger \dy \dy^\dagger w = \mu \|w\|_2^2.
\end{align}
The fact that $\dx \dx^\dagger$ is symmetric and represents the
orthogonal projection operator on $\range(\dx)$, in conjunction with
$w \in \range(\dx)$, allows us to write $w^T \dx \dx^\dagger =
w^T$. This, combined with~\eqref{eq:rayleigh-quotient-1}, yields
$\mu = \frac{w^T \dy \dy^\dagger w}{\|w\|_2^2}$.  Hence,
$\lambda_{\min}(\dy \dy^\dagger) \leq \mu \leq \lambda_{\max}(\dy
\dy^\dagger)$.
However, since $\dy \dy^\dagger$ is an orthogonal projection
operator, we have $\spec(\dy \dy^\dagger) \subset
[0,1]$. Hence, $\mu \in[0,1]$, leading
to~\eqref{eq:mspec-to-prove}, concluding the proof.
\end{proof}

From Lemma~\ref{l:consistency-matrix-properties}, the
consistency matrix is similar to a positive semidefinite matrix.  The
larger the eigenvalues of $M_C$, the more inconsistent the forward and
backward EDMD models get.  Also, from Lemma~\ref{l:consistency-matrix-properties},
$\ic = \lambda_{\max}(M_C) \in [0,1]$. Intuitively, the consistency
index determines the quality of the subspace spanned by the dictionary
and the prediction accuracy of EDMD on it. This is formalized next.

\section{Consistency Index Determines EDMD's Prediction Accuracy on
Data}\label{sec:EDMD-Prediction-Accuracy}

Our main result states that the square root of the consistency index
is a tight upper bound for the relative root mean square prediction
error of EDMD.

\begin{theorem}\longthmtitle{$\sqrt{\ic}$ Bounds the Relative Root
Mean Square Error (RRMSE) of
EDMD}\label{t:RRMSE-bound-sprad-consistency}
For dictionary $D$ and data matrices $X,Y$, under
Assumption~\ref{a:D(X)-full-rank},
\begin{align*}
\rrmsemax
&:= \max_{f \in \Span(D)} \frac{\sqrt {\frac{1}{N}
\sum_{i=1}^N | \Kc f (x_i) - \Pf_{\Kc f}(x_i)|^2 } }{ \sqrt{\frac{1}{N} \sum_{i=1}^N | \Kc f (x_i)|^2 } } 
\nonumber
\\
&=  \sqrt{\ic(D,X,Y)}.
\end{align*}
where $x_i$ is the $i$th row of $X$ and  $\Pf_{\Kc f}$ is defined
in~\eqref{eq:EDMD-Koopman-pred}.
\end{theorem}
\smallskip

Note that the combination of Definition~\ref{def:consistency-matrix}, Lemma~\ref{l:consistency-matrix-properties}, and
Theorem~\ref{t:RRMSE-bound-sprad-consistency} mean that
$\sqrt{\ic(D,X,Y)}$ satisfies all the requirements in
Problem~\ref{prob:dictionary-invariance-measure}\footnote{In
fact,
if one were to plot it as a function of
$\alpha \in \real \setminus \{0\}$ in Example~\ref{ex:motivation},
one would obtain a constant value (unlike the residual error
plotted in Figure~\ref{fig:error-vs-alpha}), showing it correctly
encodes the quality of the vector space.}.  Before proving the
result, we first remark its importance regarding function predictions.

\begin{remark}\longthmtitle{$\sqrt{\ic}$ Determines the Relative
$L_2$-norm Error of EDMD's Prediction under Empirical Measure}
Given that the elements of $\Span(D)$ and their composition with~$T$
are measurable and considering the empirical measure
$\mu_X = \frac{1}{N} \sum_{i=1}^{N} \delta_{x_i}$, where
$\delta_{x_i}$ is the Dirac measure defined based on the $i$th row
of $X$, one can rewrite $\rrmsemax$ as
\begin{align*}
\rrmsemax \! =\max_{f \in \Span(D)} \frac{\|\Kc f - \Pf_{\Kc f}
\|_{L_2(\mu_X)}}{{\| \Kc f \|_{L_2(\mu_X)}}} = \sqrt{\ic}. 
\eqoprocend
\end{align*}
\end{remark}
\smallskip

To prove Theorem~\ref{t:RRMSE-bound-sprad-consistency}, we first
provide the following alternative expression of the consistency index.

\begin{theorem}\longthmtitle{Consistency Index and Difference of
Projections}\label{t:spectral-radius-consistency-diff-proj}
Given Assumption~\ref{a:D(X)-full-rank},
\begin{align*}
\sqrt{\ic(D,X,Y)} = \sprad \big(\dy \dy^\dagger  - \dx \dx^\dagger \big).
\end{align*}
\end{theorem}
\smallskip
\begin{proof}
From Lemma~\ref{l:consistency-matrix-properties}, we have
$\ic = \lambda_{\max}(M_C)$.  We use the following notation
throughout the proof,
\begin{align*}
\lambda_{\max} =  \ic, \, \Pc_{\dx} =  \dx \dx^\dagger, \,
\Pc_{\dy} = \dy \dy^\dagger. 
\end{align*}
Note that $\Pc_{\dx}$ and $\Pc_{\dy}$ are projection operators on
$\range(\dx)$ and $\range(\dy)$, resp.  By
Definition~\ref{def:consistency-matrix}, given an eigenvalue
$\lambda \in [0,1]$ of $M_C$ with eigenvector $v \neq 0$,
\begin{align}\label{eq:consistency-fbedmd-eigs}
M_C v = \lambda v \Leftrightarrow K_F K_B v = (1-\lambda) v.
\end{align}
We consider the cases (i) $\lambda_{\max} = 0$, (ii)
$\lambda_{\max} = 1$, and (iii) $\lambda_{\max} \in (0,1)$
separately.

\textbf{Case (i): $\lambda_{\max} =0$.} In this case, from Lemma~\ref{l:consistency-matrix-properties}, we deduce $M_C = 0$.
Consequently, $K_F K_B = I$. By multiplying both sides from the left
by $\dx$ and collecting the terms, we have
$\Pc_{\dx} \Pc_{\dy} \dx = \dx$. Hence, one can write
\begin{align}\label{eq:product-proj-eigenfunctions}
\Pc_{\dx} \Pc_{\dy} z = z, \; \forall z \in \range(\dx).
\end{align}
Based on~\cite[summary table in p. 298]{WNA-EJH-GET:85}, we
deduce
\begin{align}\label{eq:product-proj-identity}
\Pc_{\dx} \Pc_{\dy} w = w \Leftrightarrow w \in \range(\dx) \cap \range(\dy).
\end{align}
Using~\eqref{eq:product-proj-eigenfunctions}-\eqref{eq:product-proj-identity},
one can write $\range(\dx) \subseteq \range(\dx) \cap \range(\dy)$ and
consequently $ \range(\dx) \subseteq \range(\dy)$.  By a similar
argument as above and swapping $K_F$ with $K_B$ and $\dx$ with $\dy$,
one can also deduce $ \range(\dy) \subseteq \range(\dx)$.  Hence,
$\range(\dx) = \range(\dy)$. Moreover, since the orthogonal projection
on a subspace is unique, we have $\Pc_{\dy}- \Pc_{\dx} = 0$, concluding
the proof for this part.

\textbf{Case (ii): $\lambda_{\max} = 1$}. By setting
$\lambda = \lambda_{\max}$ in~\eqref{eq:consistency-fbedmd-eigs},
multiplying both sides from the left by $\dx$, defining
$w := \dx v$,
we have
\begin{align}\label{eq:proj-prod-zero}
\Pc_{\dx} \Pc_{\dy} w = 0. 
\end{align}
Hence, noting that $w \neq 0$ (based on
Assumption~\ref{a:D(X)-full-rank} and the fact that $v \neq 0$), we
can deduce it is an eigenvector of $\Pc_{\dx} \Pc_{\dy}$ with
eigenvalue $0$. We show next that
$w \in \range(\dx) \cap \range(\dy)^\perp$. One can
write $\range(\dx)$ as the direct sum of the orthogonal subspaces
$\range(\dx) \cap \range(\dy)$ and
$\range(\dx) \cap \range(\dy)^\perp$. Hence, we uniquely
decompose $w \in \range(\dx)$ as $w = w_{\dy} + w_{\dy ^\perp}$,
where $w_{\dy} \in \range(\dx) \cap \range(\dy)$ and
$w_{\dy ^\perp} \in \range(\dx) \cap \range(\dy)^\perp$. Noting
that $\Pc_{\dy} w_{\dy ^\perp} = 0$ and $\Pc_{\dy} w_{\dy} = w_{\dy}$,
we get from~\eqref{eq:proj-prod-zero} that
$\Pc_{\dx} \Pc_{\dy} w = \Pc_{\dx} w_{\dy} =0$.  Since
$w_{\dy} \in \range(\dx)$, we deduce $w_{\dy} =0$ and consequently,
\begin{align*}
w = w_{\dy ^\perp} \in \range(\dx) \cap \range(\dy)^\perp.
\end{align*}
Therefore, $(\Pc_{\dy} - \Pc_{\dx}) w = - w$ and, given that
$w \neq 0$, we deduce that $\Pc_{\dy} - \Pc_{\dx}$ has an eigenvalue
equal to $-1$. Since $\spec(\Pc_{\dy} - \Pc_{\dx}) \subset [-1,1]$,
cf.~\cite[Lemma~1]{WNA-EJH-GET:85}, we conclude
$\sprad(\Pc_{\dy} - \Pc_{\dx}) = 1$.
The proof concludes by noting that $\ic = \lambda_{\max} =1$.

\textbf{Case (iii): $\lambda_{\max} \in (0,1)$.} Using
Lemma~\ref{l:invariance-spectrum-change-order} and the closed-form
expressions of $K_F$, $K_B$, $\Pc_{\dx}$, and $\Pc_{\dy}$,
\begin{align}\label{eq:fbedmd-projection-product-eigs}
\specnz(K_F K_B) = \specnz(\Pc_{\dx} \Pc_{\dy}).
\end{align}
Given $\mu \in (0,1)$, from~\cite[Theorems~1-2]{WNA-EJH-GET:85},
$\mu \in \specnz(\Pc_{\dx} \Pc_{\dy})$ if and only if
$\{\pm \sqrt{1-\mu}\} \subset \specnz\big(\Pc_{\dy} -
\Pc_{\dx}\big)$. Setting $\mu = 1- \lambda$, $\lambda \in (0,1)$,
one can use this  in conjunction
with~\eqref{eq:consistency-fbedmd-eigs}
and~\eqref{eq:fbedmd-projection-product-eigs} to write
\begin{align*}
\lambda \in \specnz (M_C) \Leftrightarrow \{\pm \sqrt{\lambda}\}
\subset \specnz\big(\Pc_{\dy} - \Pc_{\dx}\big) .
\end{align*}
This, in conjunction with~\cite[Theorem~1]{WNA-EJH-GET:85} and the
fact that $\sprad\big(\Pc_{\dy} - \Pc_{\dx}\big) \leq 1$
(cf.~\cite[Lemma~1]{WNA-EJH-GET:85}), shows that if
$\sprad\big(\Pc_{\dy} - \Pc_{\dx}\big) <1$, then the result holds.  To
conclude the proof, we need to show that
$\sprad\big(\Pc_{\dy} - \Pc_{\dx}\big) =1$ is not true. By
contradiction, suppose this is the case, then at least one of the
following holds:
\begin{enumerate}
\item
$ \exists w_1\in \real^{N} \setminus \{0\}; \; \big(\Pc_{\dy} -
\Pc_{\dx}\big) w_1 = -w_1$, \label{case:diff-proj-eig-minus-1}
\item
$\exists w_2 \in \real^{N} \setminus \{0\}; \; \big(\Pc_{\dy} -
\Pc_{\dx}\big) w_2 = w_2$. \label{case:diff-proj-eig-plus-1}
\end{enumerate}
For case~\ref{case:diff-proj-eig-minus-1}, based on~\cite[summary
table in p. 298]{WNA-EJH-GET:85},
\begin{align*}
w_1 \in \range(\dx) \cap \range(\dy)^\perp.
\end{align*}
Now, consider the vector $p_1 \neq 0$ with $w_1 = \dx
p_1$. Consequently, one can write
$K_F K_B p_1 = \dx^\dagger \Pc_{\dy} w_1 = 0$,
where in the last equality we have used $w_1 \perp
\range(\dy)$. However, this implies that $M_C p_1 = p_1$,
contradicting the fact that $\lambda_{\max} \in (0,1)$.

For case~\ref{case:diff-proj-eig-plus-1}, note that
$w_2 \in \range(\dy) \cap \range(\dx)^\perp$ (see
e.g.,~\cite[summary table in p. 298]{WNA-EJH-GET:85}).
Consider the vector space
$\Sc = \range(\dy) \cap \range(w_2)^\perp$. Clearly
$\dim \Sc < \dim \range(\dy) = \dim \range(\dx)$.
Consequently, there exists a non-zero vector $w^* \in \range(\dx)$
such that $w^* \perp \Sc$. Also, $w^* \perp \range(w_2)$ since
$w_2 \perp \range(\dx)$. Hence, by noting that $\range(\dy)$ is the
direct sum of $\Sc$ and $\range(w_2)$, one can conclude
$w^* \in \range(\dx) \cap \range(\dy)^\perp$
and, as a result, we
have $\big( \Pc_{\dy} - \Pc_{\dx} \big) w^* = -w^*$.
Since $w^*$ satisfies the identity in
case~\ref{case:diff-proj-eig-minus-1}, the proof follows by
replacing $w_1$ with $w^*$ in the proof of
case~\ref{case:diff-proj-eig-minus-1}.
\end{proof}

\begin{remark}\longthmtitle{Geometric Connections to Grassmannians}
Theorem~\ref{t:spectral-radius-consistency-diff-proj} establishes
a link between the consistency index and the spectral radius of
the difference of projection matrices. Given proper confinement of
subspaces with fixed dimension to a Grassmannian, see
e.g.~\cite{PAA-RM-RS:09}, the consistency index can be viewed as a
metric measuring the distance (by encoding angles) between vector
spaces of fixed dimension.  Similar ideas based on the difference
of projections have been used in the context of dynamic mode
decomposition~\cite{AK-TTG:21}.  Even if the dimension of the
vector spaces is not fixed, the consistency index and difference
of projections still can be used through results similar to
Theorem~\ref{t:RRMSE-bound-sprad-consistency}. This is especially
relevant if the dimension of the Koopman-invariant subspace is
unknown, see e.g.~\cite{MH-JC:21-auto}.  \oprocend
\end{remark}

We are finally ready to prove
Theorem~\ref{t:RRMSE-bound-sprad-consistency}.

\begin{proof}\longthmtitle{Theorem~\ref{t:RRMSE-bound-sprad-consistency}}
We use the following notation throughout the proof:
$ \Pc_{\dx}=\dx \dx^\dagger$, $\Pc_{\dy}= \dy\dy^\dagger$, and
$s = \sprad(\Pc_{\dy} - \Pc_{\dx})$.  Note that, from
Theorem~\ref{t:spectral-radius-consistency-diff-proj},
$s = \sqrt{\ic(D,X,Y)}$.  Given an arbitrary function
$f(\cdot) = D(\cdot) v_f \in \Span(D)$, with $v_f \in \cplx^{N_d}$,
one can use~\eqref{eq:Koopman-def}, the
predictor~\eqref{eq:EDMD-Koopman-pred} with
$K_F = \Kedmd= \EDMD{D}{X}{Y}$, and the relationship between the
rows of $X$ and $Y$ in~\eqref{eq:data-snapshots} to write
\begin{align}\label{eq:RRMSE-simplification}
\rrmse_f
&:= \frac{\sqrt { \sum_{i=1}^N |  D(y_i)v_f - D(x_i) K_F v_f|^2
} }{ \sqrt { \sum_{i=1}^N |  D(y_i) v_f|^2 }}  
\nonumber \\
&= \frac{\|\dy v_f - \dx K_F v_f\|_2}{\|\dy v_f \|_2}
\end{align}
Noting that $\dy = \Pc_{\dy} \dy$, one can write
\begin{align}\label{eq:upper-bound-sprad}
\|\dy v_f - \dx K_F v_f\|_2 
&= \|  (\Pc_{\dy}- \Pc_{\dx} ) \dy v_f\|_2 
\nonumber \\
&\leq s \|\dy v_f\|_2 ,
\end{align}
where the last inequality holds since the matrix
$\Pc_{\dy}- \Pc_{\dx}$ is symmetric and therefore its spectral
radius is equal to its induced 2-norm. Based
on~\eqref{eq:RRMSE-simplification}-\eqref{eq:upper-bound-sprad},
we have $\rrmse_f \leq s$. Hence, by definition of $\rrmsemax$ in
the statement of the result, we have
\begin{align}\label{eq:rrmsemax-inequality}
\rrmsemax = \max_{f \in \Span(D)}  \rrmse_f \leq s 
\end{align}
Now, we prove that the equality in~\eqref{eq:rrmsemax-inequality}
holds.  We consider three cases: (i) $s =0$ (ii) $s =1$ or (iii)
$s \in (0,1)$.

\textbf{Case~(i):} Since $\rrmsemax \geq 0$ by definition, in this
case $\rrmsemax = 0$ follows directly.

\textbf{Case~(ii):} In this case, there exists a
nonzero vector\footnote{The argument for the existence of $p^*$ is similar
(by swapping $\dx$ and $\dy$) to the argument used for the
existence of vectors $w_1$ and $w^*$ in the proof of
Theorem~\ref{t:spectral-radius-consistency-diff-proj}~(Case~(iii)). We
omit it for space reasons.}
$p^* \in \range(\dy) \cap \range(\dx)^\perp$. Let $v^*$ be such
that $p^* = \dy v^*$. Using~\eqref{eq:upper-bound-sprad} for $v^*$
instead of $v_f$, and the properties of $p^*$, one can write
$\| \dy v^* - \dx K_F v^*\|_2 = \|\Pc_{\dy} \dy v^*\|_2 = \|\dy
v^*\|_2$. Hence, for the function
$f^*(\cdot) = D(\cdot)v^* \in \Span(D)$, one can
use~\eqref{eq:RRMSE-simplification} to see that
$\rrmse_{f^*} = 1 =s$. Hence, equality holds 
in~\eqref{eq:rrmsemax-inequality}.

\textbf{Case~(iii):} In this case $s \in (0,1)$ and based
on~\cite[Theorem~1]{WNA-EJH-GET:85}, the matrix
$\Pc_{\dy} - \Pc_{\dx}$ has two eigenvalues $\pm s$ with
corresponding orthogonal eigenvectors
$v_{+s}, v_{-s} \in \real^{N_d}$. Moreover, based
on~\cite[Theorem~1(a)]{WNA-EJH-GET:85},
$\Pc_{\dy} v_{+s} \in \Span(v_{+s}, v_{-s})$. Hence, for some
$\alpha, \beta \in \real$, we have
\begin{align*}
q^* := \Pc_{\dy} v_{+s} = \alpha v_{+s} + \beta v_{-s}.
\end{align*}
Let $r^*$ be such that $q^* = \dy r^*$. Now, based on the first part
of~\eqref{eq:upper-bound-sprad} for $r^*$ instead of $v_f$, we have
\begin{align}\label{eq:spec-bound-pred-error}
&\| \dy r^* - \dx K_F r^* \|_2 = \| (\Pc_{\dy}-\Pc_{\dx}) \dy
r^*\|_2  
\nonumber \\
&=  \| (\Pc_{\dy}-\Pc_{\dx}) (\alpha v_{+s} + \beta v_{-s})\|_2 = s
\| \alpha v_{+s} - \beta v_{-s}\|_2 
\nonumber \\
& = s \| \alpha v_{+s} + \beta v_{-s}\|_2 =  s \| \dy r^* \|_2,
\end{align}
where
in the third and fourth equalities we have used the definition of
$v_{+s}$ and $v_{-s}$ and their orthogonality. Now, for the
function $g^*(\cdot) = D(\cdot)r^* \in \Span(D)$, one can
use~\eqref{eq:RRMSE-simplification} to see that
$\rrmse_{g^*} = s$. Hence, the equality
in~\eqref{eq:rrmsemax-inequality} holds, and this concludes the
proof. 
\end{proof}

\begin{remark}\longthmtitle{Working with Consistency Matrix is More
Efficient than the Difference of Projections}
According to Theorems~\ref{t:RRMSE-bound-sprad-consistency}
and~\ref{t:spectral-radius-consistency-diff-proj}, one can use the
consistency matrix $M_C \in \real^{N_d \times N_d}$ or the
difference of projections
$\dy \dy^\dagger - \dx \dx^\dagger \in \real^{N \times N}$
interchangeably to compute the relative root mean square
error. However, note that the size of the consistency matrix depends
on the dictionary $N_d$, while the size of the difference of
projections depends on the size of data $N$. In most practical
settings $N \gg N_d$, and consequently, working with the consistency
matrix is more efficient. In fact, given moderate to large data
sets, even saving the difference of projections matrix in the memory
may be infeasible. The calculation of the consistency matrix
requires solving two least-squares problems, which can be done
recursively for large data sets.  \oprocend
\end{remark}

\begin{remark}\longthmtitle{Efficient Computation of the Consistency Index}
The consistency index is defined as the spectral radius of the
consistency matrix and can be computed as such.  One can also use
the following to compute it more efficiently: (i) the consistency
index is the maximum eigenvalue of a matrix with nonnegative real
eigenvalues
(cf. Lemma~\ref{l:consistency-matrix-properties}(c));
(ii) given an appropriate change of coordinates making
$\dx^T \dx = I$ (see proof of Lemma~\ref{l:consistency-matrix-properties}(a)), the
consistency matrix becomes positive semi-definite. Hence, in
optimization-based dictionary learning, one can add a constraint
$\dx^T \dx = I$ and minimize the 2-norm of the consistency matrix
(which equals the maximum eigenvalue for positive semi-definite
matrices).  \oprocend
\end{remark}

\section{Conclusions}
We have introduced the concept of consistency index, a data-driven
measure that quantifies the accuracy of the EDMD method on a
finite-dimensional functional space generated by a dictionary of
functions.  The consistency index is invariant under the choice of
basis of the functional space, is computable in closed form, and
corresponds to the relative root mean squared error.  Future work will
build on the measure introduced here to design algebraic algorithms
that find dictionaries with EDMD predictions of a predetermined level
of accuracy and use the consistency index as an objective for
optimization and neural network-based methods to identify dictionaries
including the system state that span spaces that are close to being
Koopman-invariant.

\setcounter{section}{0} \renewcommand{\thesection}{Appendix
\Alph{section}}
\section{Basic Algebraic Results}\label{app1}
\renewcommand{\thesection}{\Alph{section}} 
Here, we recall two results that are used in the proofs.

\begin{lemma}\label{l:invariance-projection-change-coordinates}
Let $B_1, B_2 \in \real^{m \times n}$ be matrices such that
$\range(B_1) = \range(B_2)$. Then
$B_1 B_1^\dagger = B_2 B_2^\dagger$.  \oprocend
\end{lemma}

The proof follows from the uniqueness of the orthogonal projection
operator on a subspace.

\begin{lemma}\longthmtitle{\cite[Proposition~4.4.10]{DSB:09}}\label{l:invariance-spectrum-change-order}
Let $A \in \real^{m \times n}$ and $B \in \real^{n \times m}$. Then,
$\specnz(AB) = \specnz(BA)$.  \oprocend
\end{lemma}

\end{document}